\documentclass[conference]{IEEEtran}
\usepackage{amsmath, amssymb, amsthm}
\usepackage[utf8]{inputenc}
\usepackage{algorithm}
\usepackage[noend]{algpseudocode}
\newtheorem{theorem}{Theorem}
\newtheorem{proposition}{Proposition}
\usepackage{graphicx}
\usepackage{cite}

\title{The Tradeoff Between Coverage and Computation in Wireless Networks}
\author{Erdem Koyuncu \\ Department of Electrical and Computer Engineering, University of Illinois at Chicago }

\begin{document}
\maketitle

\begin{abstract}
We consider a distributed edge computing scenario consisting of several wireless nodes that are located over an area of interest. Specifically, some of the ``master'' nodes are tasked to sense the environment (e.g., by acquiring images or videos via cameras) and process the corresponding sensory data, while the other nodes are assigned as ``workers'' to help the computationally-intensive processing tasks of the masters. A new tradeoff that has not been previously explored in the existing literature arises in such a formulation: On one hand, one wishes to allocate as many master nodes as possible to cover a large area for accurate monitoring. On the other hand, one also wishes to allocate as many worker nodes as possible to maximize the computation rate of the sensed data. It is in the context of this tradeoff that this work is presented. By utilizing the basic physical layer principles of wireless communication systems, we formulate and analyze the tradeoff between the coverage and computation performance of spatial networks. We also present an algorithm to find the optimal tradeoff and demonstrate its performance through numerical simulations.
\end{abstract}
\section{Introduction}
In many applications, one needs to continuously monitor or cover a geographical area of interest \cite{huang2005coverage, meguerdichian2001coverage, cardei2006energy}. Objectives can be detecting significant events on the area, collecting data from various information sources such as sensors or mobile users, among other use cases. 
For example, multiple unmanned aerial vehicles (UAVs) acting as mobile base stations may be tasked to collect data from ground sensor units \cite{alzenad20173, 8038869, ekj14, ekc21}. In another scenario, the UAVs may themselves act as sensors and be tasked to collect images or video from an urban area for detecting various types of illegal activity \cite{he2018uav,zaheer2016aerial}. On the other hand, the collected data is often useless without further processing. As an example, suppose that the UAVs are tasked to detect speeding vehicles or drunk drivers on a highway using onboard cameras. The images acquired by the UAVs should be processed for determining the speed and  the motion of the vehicle, reading the license plate if possible, etc. 



Many sensing applications are time-critical, meaning that the collection of the data and its further processing should be completed with as little delay as possible. In particular, delaying the detection of a speeding vehicle or a drunk driver may have disastrous consequences as the driver of the vehicle may cause an accident unless stopped by law enforcement. Offloading high-rate sensed data to a fixed ground processing unit may incur large delays due to the distance of the unit to the wireless sensor node, and may thus not be feasible. In a remote environment, a ground processing unit may not even exist. The sensed data should then be processed by the sensors themselves, resulting in an edge computing scenario \cite{shi2016edge}.


The state-of-the-art for processing image or video for event detection or classification is to use deep neural networks, which are computationally very demanding. On the other hand, most edge devices are especially limited in terms of their power and computational capabilities. For example, a single UAV running a deep neural network or a general machine learning algorithm by itself may result in an unacceptable computational delay in the case of a time-critical application. We thus consider offloading the computationally intensive tasks of one wireless node to multiple nodes in general. 

A general survey on distributed computing schemes in wireless networks can be found in \cite{datla2012wireless}. It has been shown that agents within close proximity can form ad hoc networks tailored for cooperative computation \cite{6903299, Miluzzo:2012:VMC:2307849.2307854}. Load balancing  on such ``transient clouds'' have also been studied \cite{7037232}. A fundamental problem in this context is to optimally allocate the available computing tasks and wireless resources to different nodes \cite{mao2017joint, xing2018joint, guo2017energy, guo2018mobile, zhou2019computation}. The effect of interference on optimal resource allocation have been studied in \cite{wang2017joint}. Applications of wireless edge computing to content caching \cite{wang2017computation}, augmented reality applications \cite{al2017energy}, and UAV networks \cite{zhou2018computation} are also available. 

The key observation of the present work is that the two goals of achieving high coverage and low computation delay work against each other. In fact, achieving high coverage requires acquiring and processing the coverage data from as many nodes as possible. On the other hand, achieving a low computation delay is only possible by discarding part of the coverage data, effectively utilizing the computational resources of the network to process the remaining data faster. Our goal is to analyze the corresponding tradeoff between coverage and computation delay. In this work, we focus on minimizing the average computation delay, as opposed to goal of minimizing the maximum computation delay over all tasks. As also summarized above, there are many existing studies that analyze the coverage and computation performances of networks individually. On the other hand, to the best of our knowledge, the tradeoff between these two important figures of merit has not been previously identified or analyzed. 


The rest of this paper is organized as follows: In Section \ref{secsystemmodel}, we introduce the system model. In Section \ref{sechardness}, we prove that finding the optimal tradeoff is NP-complete. In Section \ref{secalgotofindclusters}, we introduce a low-complexity algorithm that provides a locally optimal solution. In Section \ref{secnumerical}, we present numerical results. Finally, in Section \ref{secconclusions}, we draw our main conclusions.

\section{System Model}
\label{secsystemmodel}
We consider $n$ nodes with locations $u_1,\ldots,u_n \in A$, where $A\subset\mathbb{R}^d$ is an area of interest, and $d\in\{1,2,3\}$ is the ambient dimension. The cases $d=1$, $d=2$, and $d=3$ may correspond to cars on a straight highway, mobile sensors on the ground, or unmanned aerial vehicles (UAVs) on the air, respectively. 

Each node can cover or survey any location that is within a distance $D$ to itself. In other words, Node $i$ can survey its coverage area $\{x:\|x - u_i\| \leq D\}$. Also, each node acquires computation tasks from its coverage area at a rate of $R_T$ tasks per second. As an example, suppose that each node is equipped with a camera that acquires video frames from its coverage area at a rate of $30$ frames per second. We would like to pass each frame through a neural network for classification purposes, e.g. for the purpose of detecting some significant event. We can then declare that each node acquires $R_T = 30$ tasks per second. Each task corresponds to feeding a video frame to a neural network and obtaining the output. One can also define one task as processing one second of video. In this case, each node would acquire one task per second.

We consider a scenario where we only process the tasks of a certain subset $M\subset \{1,\ldots,n\}$ of nodes that we shall refer to as ``master nodes.'' Also, Master Node $i$, where $i\in M$, is assigned a set of worker nodes $W_i\subset\{1,\ldots,n\}$ to help processing the tasks of the master. We assume the sets $M,W_1,\ldots,W_n$ are disjoint. We refer to the set $C_i \triangleq \{i\} \cup W_i$ of Master Node $i$ and its workers as Cluster $i$.

\begin{figure}[h]
\scalebox{0.42}{\includegraphics{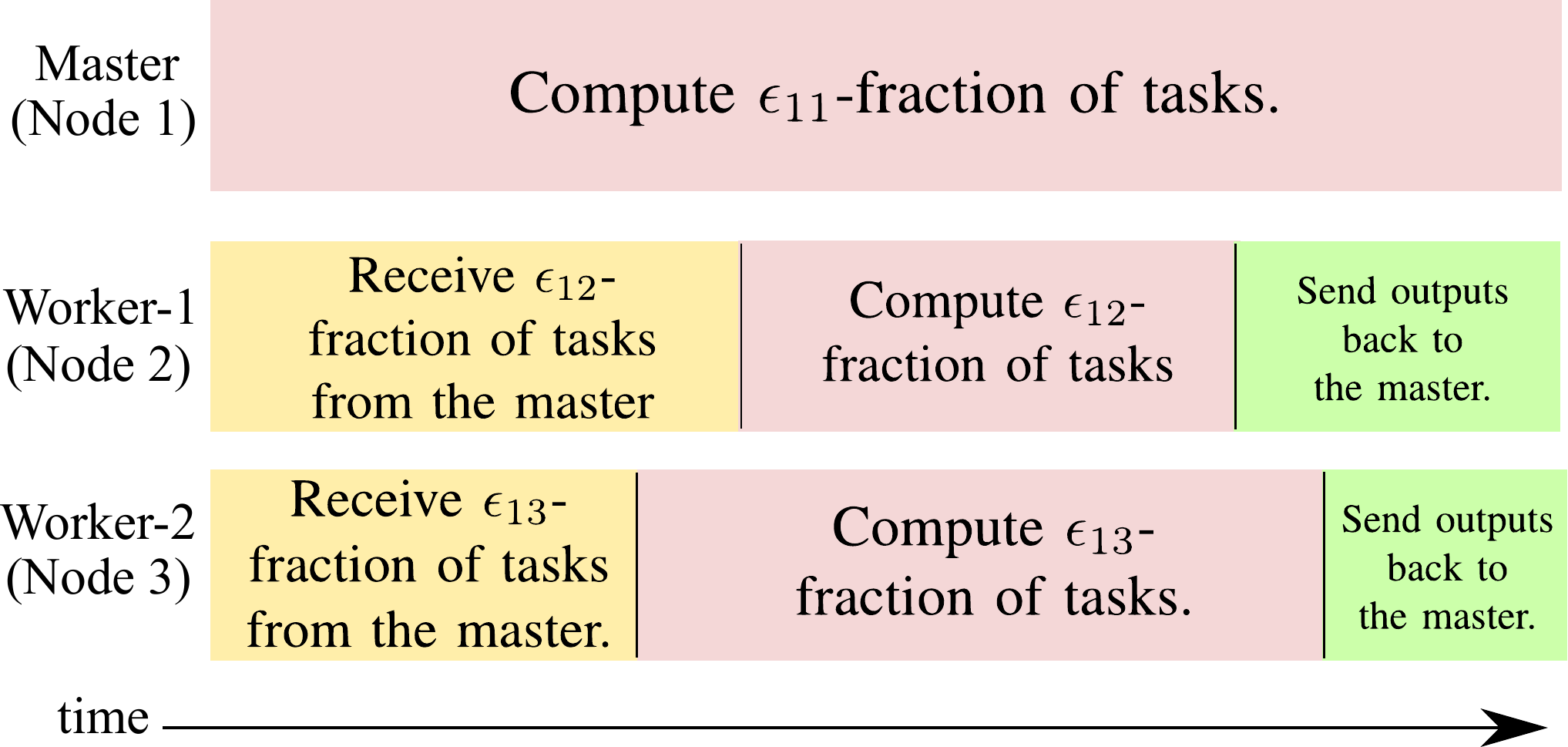}}
\caption{A possible frame for Cluster $1$ with workers $W_1 = \{2,3\}$, under the assumption that all nodes within the cluster finish their computation and communication jobs at the same time. In fact, Proposition \ref{propopttaskassign} of Section \ref{sechardness} will show that this assumption is optimal.}
\label{computcommunframe}
\end{figure}

One example schedule for communication and computing at a given cluster of master node and its workers is illustrated in Fig. \ref{computcommunframe}. Time is divided into disjoint frames. At each frame, $T$ tasks should be processed at each master node, resulting in a total of $T|M|$ tasks to process at each frame. Suppose that Master Node $i$ processes $\epsilon_{ii}$-fraction of the tasks, while Worker Node $j$, where $j\in W_i$, processes $\epsilon_{ij}$-fraction of the tasks.\footnote{In general, we allow $\epsilon_{ij}$s to be arbitrary real numbers, in which case $\epsilon_{ij}T$ is not necessarily an integer. To deal with this technicality, one can consider the $T\rightarrow\infty$ regime, in which case all the computation rate or coverage results presented in the paper will be asymptotically achievable.} In other words, within Cluster $i$, Node $j$ processes $\epsilon_{ij}$-fraction of the tasks, where $j\in C_i$. We have the constraint $\sum_{j\in C_i} \epsilon_{ij} = 1$ for every $i\in M$. At the beginning of each frame, the master node begins calculating its $\epsilon_{ii}T$ tasks. Let $\gamma_i$ denote the processing speed of Node $i$, measured in tasks per second. It thus takes $\frac{\epsilon_{ii}T}{\gamma_i}$ seconds for Master Node $i$ to finish processing its tasks. 

At the beginning of each frame, the master node also sends the individual task data to each one of its workers. Specifically, given $i\in\{1,\ldots,n\}$ and $j\in W_i$, Master Node $i$ sends the data for $\epsilon_{ij} T$ tasks to Worker Node $j$. Assuming that the length of the data to be processed at each task is $b_0$ bits, Master Node $i$ thus sends $\epsilon_{ij} T b_0$ bits to Worker Node $j$. Now, let the rate of data communication between Nodes $i$ and $j$ of the network is given by $\rho_{ij}$. The transmission of the task data from Master Node $i$ to Worker Node $j$ then takes $\epsilon_{ij} T b_0/\rho_{ij}$ seconds. We assume orthogonal channels between a master and its worker nodes so that transmissions from a master node to its workers occur simultaneously. We leave  non-orthogonal access as an interesting direction for future work. 

Let us now discuss the nature of data rates between nodes in more detail. Most of the conclusions of the paper can be generalized to different models that specify inter-node communication rates $\rho_{ij}$s. We will, however, always assume that the reliable communication rate between two nodes increases as the distance between them decreases. In particular, $\rho_{ij}$ should be expressed as $\rho_{ij} = \rho(\|u_i - u_j\|)$, where $\rho(\cdot)$ is a monotonically non-increasing function. Also, if two nodes are within a finite distance, then a non-zero communication rate between them should be achievable; i.e., $d > 0 \implies \rho(d) > 0$.


Going back to our communication and computing scheme, as soon as a worker node receives all its task data, it begins its computations.\footnote{Therefore, in our model, a worker cannot begin any computation unless it receives all its tasks. This is to take into account a possible joint compression of the tasks at the master node side. For example, if each task is a video frame to be processed, then multiple frames can be coded together at the master node and then transmitted. Decoding at the worker is possible only when the worker receives data corresponding to multiple frames.}
 Let us recall that $\gamma_i$ denotes the processing speed of Node $i$, measured in tasks per second.
It then takes $\frac{\epsilon_{ij} T}{\gamma_j}$ seconds for Worker Node $j$ to process all its tasks. 

As soon as a worker node finishes computing all its tasks, it sends the corresponding outputs back to the master node. To calculate the transmission times of these outputs, suppose that processing each task results in an output of $b_1$ bits. Then, given $i\in\{1,\ldots,n\}$ and $j\in W_i$, Worker Node $j$  sends $\epsilon_{ij} T b_1$ bits back to Master Node $i$. In this work, we assume symmetric links so that the capacities of the forward and backward links between two nodes are the same. In other words, $\rho_{ij} = \rho_{ji},\,\forall i,j$. Hence, the transmission from Worker Node $j$ to Master Node $i$ takes $\frac{\epsilon_{ij} T b_1}{\rho_{ij}}$ seconds.


It follows from the above transmission/computing scheme that the cluster led by Master Node $i$ is completed within
\begin{align}
    \tau_i \triangleq \max_{j\in\{i\}\cup W_i}\left\{\frac{\epsilon_{ij} T b_0}{\rho_{ij}} + \frac{\epsilon_{ij} T}{\gamma_j} +  \frac{\epsilon_{ij} T b_1}{\rho_{ij}}\right\} \mbox{\,\,seconds,}
\end{align}
with the convention that $\rho_{ii} = \infty$.  Equivalently, the cluster can process $R_i \triangleq \frac{T}{\tau_i}$ tasks per second. 

Note that only the data acquired from the coverage area of the master nodes are processed. The total coverage area provided by the master nodes are given by
\begin{align}
\label{coverageareamaster}
    C \triangleq \left| \{x:\exists i\in M ,\,\| x - u_i\| \leq D\} \right|,
\end{align}
where $|\cdot|$ denotes the $d$-dimensional volume of a set. The sensing information acquired from the coverage area is processed with a computation rate of
\begin{align}
   R \triangleq \min_{i\in M} R_i \,\,\mbox{tasks}/\mbox{second}
\end{align}
or greater. The computation rate $R$ of the network has the following interpretation: Recall that each vehicle acquires tasks from the environment at a rate of $R_T$ tasks per second. If $R \geq R_T$, then the network is stable in the sense that the number of unprocessed tasks is bounded by a constant at all times. However, if $R < R_T$, the number of unprocessed tasks grows to infinity with time, resulting in an unstable network.

In this work, we are interested in determining the tradeoff between the computation rate $R$ and the coverage $C$. In other words, we would like to determine the tradeoff 
\begin{align}
\label{tradeoffform}
R^{\star}(c) \triangleq \max\{R : C \geq c\}
\end{align}
for different values of the coverage $c$. From now on, to conveniently present our analytical results, we set $b_0 = b_1 = \frac{1}{2}$ without loss of generality (Later in Section \ref{secnumerical}, we will use practical values for all parameters). In this case, the computation rate can be expressed as
\begin{align}
\label{comprateopt}
    R = 1 / \max_{i\in M} \max_{j\in C_i}\epsilon_{ij}\alpha_{ij},
\end{align}
where
\begin{align}
\alpha_{ij}\triangleq  \frac{1}{\rho_{ij}} + \frac{1}{\gamma_j}.
\end{align}
Note that an equivalent problem is to investigate the tradeoff between coverage and the per-task computation delay $\tau \triangleq \frac{1}{R}$.


According to our formulation so far, if a particular location on the area of interest is covered, say, $N > 1$ times by $N$ master nodes, then the corresponding data acquired from this location is also processed $N > 1$ times. While processing the same data multiple times seems like a waste of computational resources, it may also be unavoidable in many practical settings. For example, suppose that the task is to acquire an image of the coverage area and feed it to a neural network. If a master node knows that a part $E$ of its coverage area is already covered by another master node, it may choose not to process the pixels acquired from $E$, and leave the processing of $E$ to other masters. This would however require more coordination among nodes (as they also need to coordinate which master should process which locations), and a sophisticated neural network that accepts an arbitrarily sized and shaped image input. On the other hand, if each node of the network corresponds to (say) an access/processing point, and if each location on the area of interest corresponds to a sensor that provides sensory data to these access points, then it may make perfect sense to process a given sensor's data only once. Extensions of our results in this direction will be discussed elsewhere.


\section{Hardness of Finding the Optimal Tradeoff}
\label{sechardness}
We first seek to determine the tradeoff (\ref{tradeoffform}) exactly. It is easily observed that the coverage $C$ is independent of the fractions of tasks $\epsilon_{ij}$s allocated among masters and workers. We can thus freely optimize the computation rate (\ref{comprateopt}) over $\epsilon_{ij}$s. The following proposition performs this optimization.

\begin{proposition}
\label{propopttaskassign}
An optimal choice for the task assignments $\epsilon_{ij},\,i\in M,\,j\in C_i$ that maximize (\ref{comprateopt}) are given by
\begin{align}
\label{opttaskassignments}
    \epsilon_{ij} = \frac{1}{\alpha_{ij}} / \sum_{j\in C_i} \frac{1}{\alpha_{ij}}.
\end{align}
The corresponding computation rate is
\begin{align}
\label{comprateoptoptxxx}
    R' = \min_{i\in M} \sum_{j\in C_i} \frac{1}{\alpha_{ij}}.
\end{align}
\end{proposition}
\begin{proof}  Maximizing the computing rate is equivalent to minimizing the computation delay $\tau = \frac{1}{R}$. We have
\begin{align}
   \min_{\epsilon_{k\ell},k\in M,\,\ell\in C_k} \tau & = \min_{\epsilon_{k\ell},k\in M,\,\ell\in C_k} \max_{i\in M} \max_{j\in C_i}\epsilon_{ij}\alpha_{ij}  \\ & \geq  \max_{i\in M} \min_{\epsilon_{k\ell},k\in M,\,\ell\in C_k} \max_{j\in C_i}\epsilon_{ij} \alpha_{ij} \\
    & \label{achlowbound} = \max_{i\in M} \min_{\epsilon_{i\ell},\,\ell\in C_i} \max_{j\in C_i}\epsilon_{ij} \alpha_{ij}
\end{align}
Consider now the minimization $\min_{\epsilon_{i\ell},\,\ell\in C_i} \max_{j\in C_i}\epsilon_{ij} \alpha_{ij}$ that appears in the final equality, where we have the inherent constraint $\sum_{\ell\in C_i} \epsilon_{i\ell}= 1$. It is easily seen that there is an optimal task assignment that satisfies $\epsilon_{ij}\alpha_{ij} = \epsilon_{ik}\alpha_{ik},\,\forall j,k$ (Otherwise, if $\epsilon_{ij}\alpha_{ij} > \epsilon_{ik}\alpha_{ik}$, one can consider the assignments $\epsilon_{ij} \leftarrow \epsilon_{ij} - \delta, \epsilon_{ik} \leftarrow \epsilon_{ik} + \delta$, where $\delta$ satisfies $(\epsilon_{ij} - \delta)\alpha_{ij} = (\epsilon_{ik} - \delta)\alpha_{ik}$.). Using the constraint $\sum_{j\in C_i} \epsilon_{ij} = 1$, we obtain the task assignments in (\ref{opttaskassignments}). Substituting the assignments in (\ref{opttaskassignments}) to (\ref{comprateopt}), we obtain the computation rate in (\ref{comprateoptoptxxx}). This concludes the proof. 
\end{proof}

From now on, we assume optimal task assignments at each cluster as indicated by Proposition \ref{propopttaskassign}. We thus seek the tradeoff between the computation rate (\ref{comprateoptoptxxx}) with optimal task assignments and the coverage (\ref{coverageareamaster}). In particular, the tradeoff in (\ref{tradeoffform}) can be expressed as
\begin{align}
\label{tradeoffform2}
R^{\star}(c) = \max\{R' : C \geq c\}.
\end{align}
To understand how difficult the problem in (\ref{tradeoffform2}) is exactly, we consider its following decision version:
\begin{align}
\label{decisionproblem}
\mbox{Given $\rho, c$, is there a clustering with }R' \geq \rho, C \geq c? 
\end{align}

We have the following result.
\begin{theorem}
\label{theorem1}
The problem (\ref{decisionproblem}) is NP-complete.
\end{theorem}
\begin{proof}
Let us first show that the problem (\ref{decisionproblem}) is in NP. In fact, since $|M| \leq n$ and $|C_i| \leq n,\,\forall i$, the calculation of $R'$ in (\ref{comprateoptoptxxx}) can be accomplished in linear time. Calculating $C$ is equivalent to the problem of calculating the union of $n$ disks, which can be accomplished in $O(n^2)$ time \cite{avis1988computing}. Thus, the problem (\ref{decisionproblem}) is in NP.

We prove the NP-completeness of (\ref{decisionproblem}) via a reduction of the so-called minimum unit-disk cover problem (MUDC) \cite{ko2011complexity}. Consider points $v_1,\ldots,v_{n'} \in \mathbb{R}^d$. The MUDC problem asks, given $k' \geq 1$, whether there is a set of indices $K' \subset \{1,\ldots,n\}$ of cardinality at most $|K'|\leq k'$ such that 
\begin{align}
\bigcup_{i\in K'} \{x:\|v_i - x\| \leq D\} = \bigcup_{i=1}^n \{x:\|v_i - x\| \leq D\}.
\end{align}
The MUDC problem has been shown to be NP-complete \cite{ko2011complexity}.

Consider, now, an instance of the MUDC problem, parameterized by the coordinates $v_1,\ldots,v_{n'}$ and $k'$. We will construct an instance of the problem in (\ref{decisionproblem}) in polynomial time; this particular instance of (\ref{decisionproblem}) will be equivalent to the instance of the MUDC problem, proving the NP-complenetess of (\ref{decisionproblem}). Let $n = 2k'$, $u_i = v_i,\,i=1,\ldots,n'$, and $u_i = v_{n'},\,i=n'+1,\ldots,n$. Let $c' = |\{x:\exists i\in\{1,\ldots,n'\}:\|x - v_i\| \leq D\}|$ denote the covering provided by $v_i$s. Note that $c'$ can be calculated in polynomial time \cite{avis1988computing}. Let $\gamma_i = 1,\,\forall i$, and define 
\begin{align}
\label{pqowepoqiweqwe}
\epsilon \triangleq \min_{i,j\in\{1,\ldots,n\}} \frac{1}{1+\frac{1}{\rho_{ij}}} > 0. 
\end{align}
Note that $\epsilon$ can also be calculated in polynomial time. We can thus consider the following instance of (\ref{decisionproblem}).
\begin{align}
\label{decisionproblemxx22}
\mbox{Is there a clustering with }R' \geq 1+\epsilon, C \geq c'? 
\end{align}
Suppose now that the MUDC problem is feasible, i.e., there exists a cardinality-$k'$ covering of $n'$ disks of equal radius. Let $K' \subset \{1,\ldots,n'\}$ denote the set of nodes in such a covering. In order to construct a solution to (\ref{decisionproblemxx22}), we declare the nodes in $K'$ as master nodes, and assign at least one worker node to each master node. Since $|K'| \leq k'$, and we have $n = 2k'$ nodes available, such a clustering is feasible. To estimate the corresponding computation rate, note that, as a result of the choice $\gamma_i = 1,\,\forall i$, we have $\alpha_{ii} = 1$ whenever $i$ corresponds to a master node. Also, it follows from (\ref{pqowepoqiweqwe}) that $\frac{1}{\alpha_{ij}} \geq \epsilon$ for every $i$ and $j$. Therefore, according to (\ref{comprateoptoptxxx}), the resulting computation rate is at least $1+\epsilon$. The clustering also clearly provides a covering of $c'$. Thus, (\ref{decisionproblemxx22}) is satisfied.

Conversely, suppose that there is a clustering that satisfies (\ref{decisionproblemxx22}). Since each master node can provide a computation rate of at most $1$ task per second, but since the overall computation rate is greater than $1$, each master node should have at least one worker. Removing these workers leaves us with a set of master nodes that provides a coverage of $c'$ with cardinality at most $k'$. Such a set is a solution to the MUDC problem.

The arguments above show that there is a solution to MUDC problem if and only if there is a solution to the instance (\ref{decisionproblemxx22}) of (\ref{decisionproblem}). This concludes the proof that (\ref{decisionproblem}) is NP-complete.
\end{proof}
Theorem \ref{theorem1} shows that finding the optimal clustering of nodes and the corresponding optimal tradeoff is a computationally hopeless problem in general. In the next section, we will seek to develop an efficient algorithm to find good clusterings.

\section{An Algorithm to Find Good Clusterings}
\label{secalgotofindclusters}
We follow a Lagrangian approach in order to design a computationally-efficient algorithm that finds good clusterings. Specifically, we set our objective to maximize the Lagrangian
\begin{align}
L \triangleq R' + \lambda C
\end{align}
over all clusterings, where $\lambda > 0$ is a parameter that allows travel over the $R'$-$C$ trade-off curve. In fact, a larger $\lambda$ translates to a larger emphasis on the coverage performance of the network, resulting in a potentially low computation rate. On the other hand, optimizing for a small $\lambda$ will provide a high computation rate but low coverage. 

We can now consider the algorithm whose pseudocode is shown in Algorithm \ref{algo1}. The algorithm is initialized with $M = \{1,\ldots,n\}$, and $W_i = \emptyset,\,\forall i\in M$. It then repeats the $4$ steps in Lines 3-6 until convergence. In the following, we provide a precise description of each step. In Line 3, we attempt to merge clusters. Specifically, let $M$ denote the existing masters with workers $W_i,\,i\in M$. For every $i<j$ with $i,j\in M$, we check whether the new clustering with masters $M' = M - \{j\}$ and workers $W_i' = W_i \cup \{j\} \cup W_j,W_k' = W_k,\,k \notin\{i,j\}$ provides a better or equal Lagrangian as compared to the existing clustering. If so, we update the master and worker sets to $M'$ and $W_i'$s, respectively. 

\begin{algorithm}
\caption{A Descent Algorithm to Find Good Clusterings}
\begin{algorithmic}[1]
\State Set all nodes as masters and none of the nodes as workers. 
\State {\bf until} convergence of the Lagrangian {\bf do}
\State \hskip1.5em Go through all pairs of clusters, and merge two clusters when it will not increase the Lagrangian.
\State \hskip1.5em Go through all clusters, and find the best master node for each cluster.
\State \hskip1.5em Assign each worker node to an optimal cluster.
\State \hskip1.5em Go through all pairs of workers, and switch the workers' clusters if it will not increase the computation rate.
\end{algorithmic}
\label{algo1}
\end{algorithm}

In Line 4, for each cluster, we find the best choice for the master node among all nodes within the cluster. While finding the best master of a cluster, we keep all other clusters fixed. Mathematically, for every $i\in M$, we calculate the Lagrangians of the topologies $M'' = M - \{i\} \cup \{j\}$, $W_i'' = W_i \cup \{i\} - \{j\}$, $W_k'' = W_k,\,k \neq i$ for different $j \in \{i\}\cup W_i$. The index $j$ that maximizes the Lagrangian replaces $i$ as the new master of the cluster. Note that, changing the master node of any given cluster will not change the computation rates of other clusters. When calculating the Lagrangians with different candidate masters, one thus has to calculate the computation rates of the other clusters only once. Exploiting this observation yields faster execution times for Algorithm \ref{algo1}.

In Line 5, we assign each worker to an optimal cluster, while keeping all master nodes fixed. Since the master nodes are fixed, so is the coverage provided by the network topology. The optimal cluster for a given worker should thus maximize the computation rate. Likewise, Line 6 goes through all pairs of workers, and switches the clusters of two pairs of workers whenever it will improve the computation rate.

It is instructive to note that Algorithm \ref{algo1} is analogous to the $k$-means algorithm. In fact, finding the best master at each cluster is analogous to finding the centroids of clusters in the $k$-means algorithm. Likewise, assigning each worker to its optimal cluster is similar to assigning a data point to its optimal cluster. The algorithm is initialized with as many clusters as the number of nodes and then merges the clusters if necessary to converge to an optimal number of clusters. This corresponds to finding the optimal ``$k$'' in the $k$-means algorithm. One major difference as compared to $k$-means is that in $k$-means, the optimal cluster of a given data point is independent of the clusters of other data points. Therefore, in $k$-means, a step like Line 6, which would consider pairs of data points, would be unnecessary. On the other hand, for our problem, the optimal cluster for a given worker depends on the clusters of other workers. Hence, we have also added Line 6 to ensure that our algorithm avoids bad local maximums. 




We now analyze the convergence and computational complexity Algorithm \ref{algo1} through the following theorem.
\begin{theorem}
With Algorithm \ref{algo1}, the Lagrangian converges to a local maximum after finitely many iterations. Moreover, the algorithm takes $O(\ell  n^4)$ time, where $\ell$ is the number of iterations performed until convergence.
\end{theorem}
\begin{proof}
Algorithm \ref{algo1} provides a monotonically non-decreasing sequence of Lagrangians. Moreover, the Lagrangian is bounded above by the absolute constant $L \leq \sum_{i=1}^n \gamma_i$ representing the aggregate computing power of all nodes in the network. The Lagrangian will thus converge according to the monotone convergence theorem. To show that convergence occurs after finitely many iterations, note that there can only be a finite number of distinct Lagrangians corresponding to a finite number network topologies. Correspondingly, we may define the largest and the second largest Lagrangian values that appear in an instance of Algorithm \ref{algo1} as $L'$ and $L''$, respectively. With this notation, the Lagrangians provided by Algorithm \ref{algo1} thus converges to $L'$. Equivalently, for every $\epsilon > 0$, the Lagrangians should be $\epsilon$-close to $L'$ at all iterations with sufficiently high indices. Choosing $\epsilon = \frac{1}{2}(L' - L'')$, the Lagrangian should equal $L'$ at all sufficiently large indices, proving convergence after finitely many iterations. At convergence, the solution is necessarily a local maximum with respect to the improvements in Lines 3-6.

To prove the computational complexity, we note that each one of the Lines 3-6 of Algorithm \ref{algo1} can be accomplished in $O(n^4)$ time. In particular, going through all pairs of clusters takes $O(n^2)$ as there are at most $\binom{n}{2} = \frac{n(n-1)}{2}$ clusters. Calculating a Lagrangian of a network topology can be accomplished in $O(n^2)$ time, as discussed in the proof of Theorem \ref{theorem1}. Similarly, it can be shown that all remaining steps can be accomplished in $O(n^4)$. This concludes the proof.
\end{proof}
An open problem is to find an upper bound on $\ell$ as a function of $n$ for a more precise description of the computational complexity. Our numerical simulations suggest that $\ell$ is sublinear in $n$, and Algorithm \ref{algo1} converges very fast. Hence, Algorithm \ref{algo1} is very likely a polynomial time algorithm with  $O(n^5)$ worst case time complexity as opposed to the exponential time needed for exhaustive search.

\section{Numerical Results}
\label{secnumerical}
In this section, we provide numerical simulation results that show the performance of our algorithms. We first describe a practical setting to set the various different parameters of the nodes as described in Section \ref{secsystemmodel}. 


\subsection{The UAV Surveillance Scenario}
\label{secscenario}
We survey a square area with side length $10$ km by multiple UAVs that are on the same altitude. In this case, $u_1,\ldots,u_n$ represent the projected ground locations of the UAVs.

Each UAV is equipped with a video camera that acquires video frames over the UAV's visual line of sight. Clearly, the coverage radius $D$ provided by each UAV will depend on the specific application, the resolution of the video camera, and the processing algorithm that processes the video frames. For example, on a clear day, a UAV flying at an altitude of a few hundred meters off the ground will be able to detect a forest smoke or fire that is several kilometers away. However, if the task is to detect criminal activities on the ground, the coverage radius will obviously be much smaller. Here, we consider an application with similar requirements as smoke or fire detection, and thus consider a coverage radius of $2$ km. 

Since the UAVs will be high up on the air, the channels between them can be considered to be free-space path loss channels with a certain exponent $r$. Suppose that the communication bandwidth is $B$ Hz, the carrier wavelength is $\lambda_c$ m,  each UAV transmits with power $P$ Watts, and the  noise power spectral density is $N_0$ Watts/Hz. The achieable rate within a distance $d$ to any of the UAVs can then be modeled by \cite{goldsmith2005wireless}
\begin{align}
\rho(d) \triangleq B\log_2\left(1\!+\! \frac{P}{B N_0} \Bigl(\frac{\lambda_c}{4\pi d_0}\Bigr)^2\left(\frac{d_0}{d}\right)^{r} \right) \mathrm{bits}/\mathrm{sec},\!
\end{align}
where $r$ is the path loss exponent, and $d_0$ depends on the environment. For outdoor attenuation, a typical value is $d_0 = 10$ \cite{goldsmith2005wireless}. Also, we consider $\lambda_c = \frac{1}{3}$ (corresponding to a carrier frequency of $900$ Mhz), $B = 1$ MHz, $P = 0$ dBm, and $N_0 = -170$ dBm. We will present our results for  $r\in\{2,2.5,3\}$.


We now derive a typical practical value for the processing speeds $\gamma_i$s at the UAVs. Suppose that the UAVs acquire high definition video at 720p quality. Typically, 720p videos have a frame rate of $30$ frames/second coded at $4$ Mbits/second. We consider a scenario where each frame of the video is fed to a neural network for further processing. For a  $224 \times 224$ input image, the typical inference times for a state-of-the-art deep neural network such as ResNet-50, which can also be used for smoke or fire detection \cite{sharma2017deep}, ranges from $1$ms on a powerful Tesla V100 GPU \cite{nvidiagpu} to $25$ms \cite{niu201926ms} on a smartphone. We assume a $10$ ms inference time. Given that the frame sizes for a 720p video is $1280 \times 720$, the $10$ ms inference times will be scaled by a factor of roughly $\frac{1280 \times 720}{224 \times 224} \approx 18$, resulting in an inference time of $180$ms per frame. A second of video is then processed within $30\times 180\mbox{ms} = 5.4$s. Defining a ``task'' as the complete processing of one second of video, we may thus set $\gamma_i = \frac{1}{5.4},\,\forall i$. Also, input to each task is of length $b_0 = 4$Mbits (owing to $4$Mbits/sec video data rate). Assuming that the output of the neural network is a basic binary decision (e.g. whether there is a smoker or not), $b_1$ will be negligible as compared to $b_0$. We thus simply assume $b_1 = 0$.

\subsection{Results}
\label{secactualresults}

With a fixed set of network parameters as described in Section \ref{secscenario}, and $50$ UAVs with fixed locations (we will indicate the specific UAV locations later on), we have run Algorithm \ref{algo1} for different values of the Lagrange multiplier $\lambda$. This gave us a set of points $S_1$ on the coverage vs. computation rate space. We then computed the Pareto frontier of these set of points by removing a coverage-rate pair $(C_1,R_1)$ in $S_1$ if there exists another pair $(C_2,R_2) \in S_1$  such that $C_2 \geq C_1$ and $R_2 \geq R_1$. The resulting Pareto frontiers for different path loss exponents are shown in Fig. \ref{paretofig}. The horizontal axis represents the fraction of total area covered, and the vertical axis represents the computation rate in tasks per second. 

\begin{figure}[htbp]
\scalebox{0.63}{\includegraphics{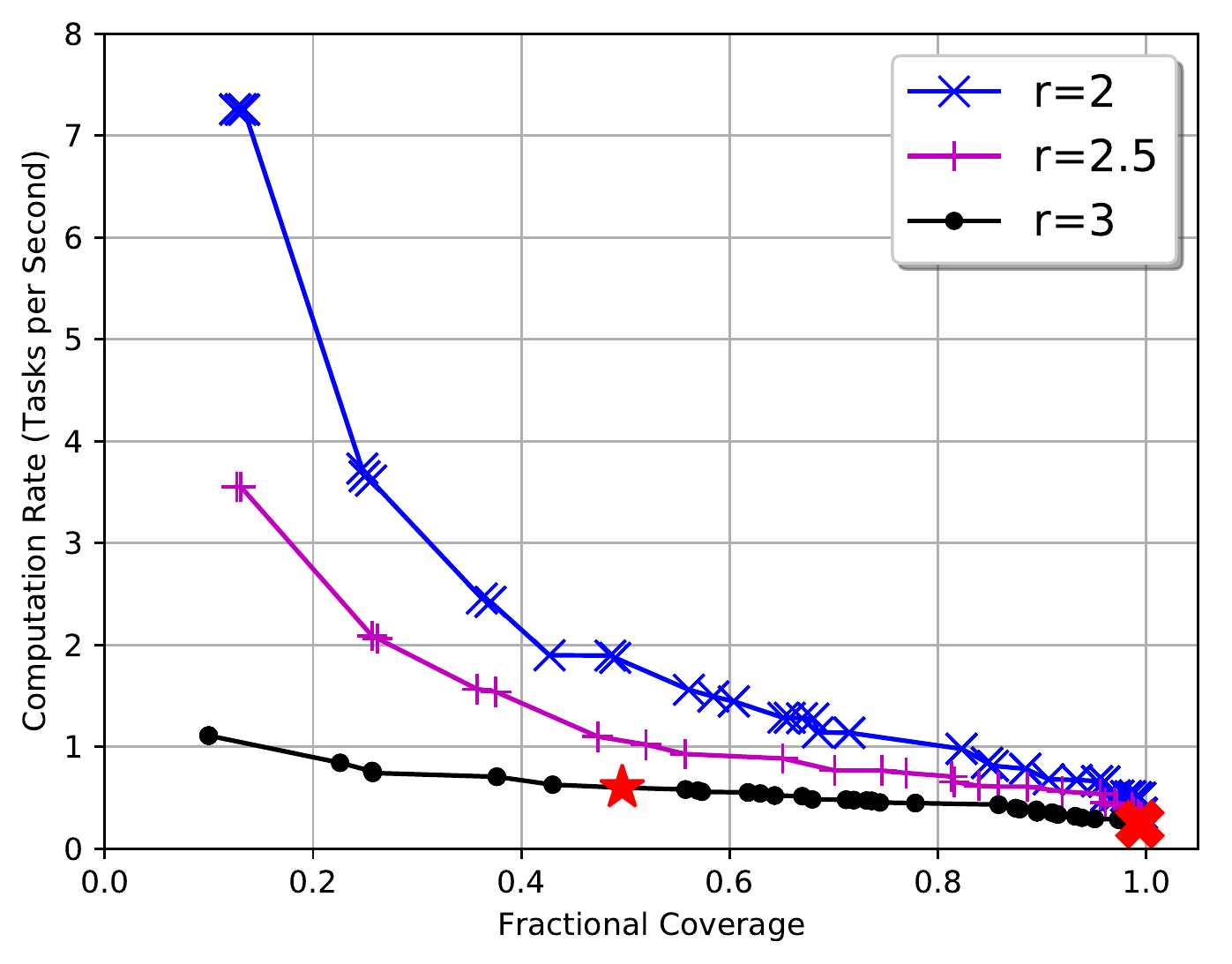}}\vspace{-13pt}
\caption{Pareto frontiers for different path loss exponents $r$.}
\label{paretofig}
\end{figure}

We can observe that the tradeoff improves as we consider a smaller path loss exponent. This is due to increased inter-UAV communication rates, which provide lower computation delays. Let us also recall from Section \ref{secscenario} that we have defined one task as the processing of one second of video through a neural network. Processing all incoming tasks is thus feasible whenever the computation rate is greater than $1$ task per second. We can observe that processing the video stream in real time becomes infeasible for most coverage constraints when the path loss exponent is $3$. In such a scenario, one can decrease the video frame rate to ensure real time inference.

For each Pareto frontier in Fig. \ref{paretofig}, the left-most point corresponds to network topologies with one master node and $49$ worker nodes connected to the unique master. As one travels to the right hand side of each curve, the number of masters increase, effectively increasing coverage but decreasing the computation rate. Two example topologies on the Pareto frontier for $r=3$ are illustrated in Figs. \ref{conf1fig} and \ref{conf2fig}. In particular, Fig. \ref{conf1fig} illustrates the topology that achieves a coverage of $0.4968$ with computation rate $0.5985$; this point is also marked with a red star in Fig. \ref{paretofig}. In Fig. \ref{conf1fig}, each master node is marked with a red square, while each worker node is marked with a blue disk. Each cluster is represented by multiple workers connected to the master via straight lines. Note that these node locations remain fixed for all the points obtained in Fig. \ref{paretofig}. We can observe from Fig. \ref{conf1fig} that masters sometimes make very long connections, and not every worker node is connected to its nearest master. In fact, such a nearest-master strategy can be observed to be suboptimal in general. For example, applying a nearest-master topology to Fig. \ref{conf1fig} would result in the top-right master node losing many of her workers without gaining any new, resulting in a lower overall computation rate.

\begin{figure}[htbp]
\begin{center}\scalebox{0.6}{\includegraphics{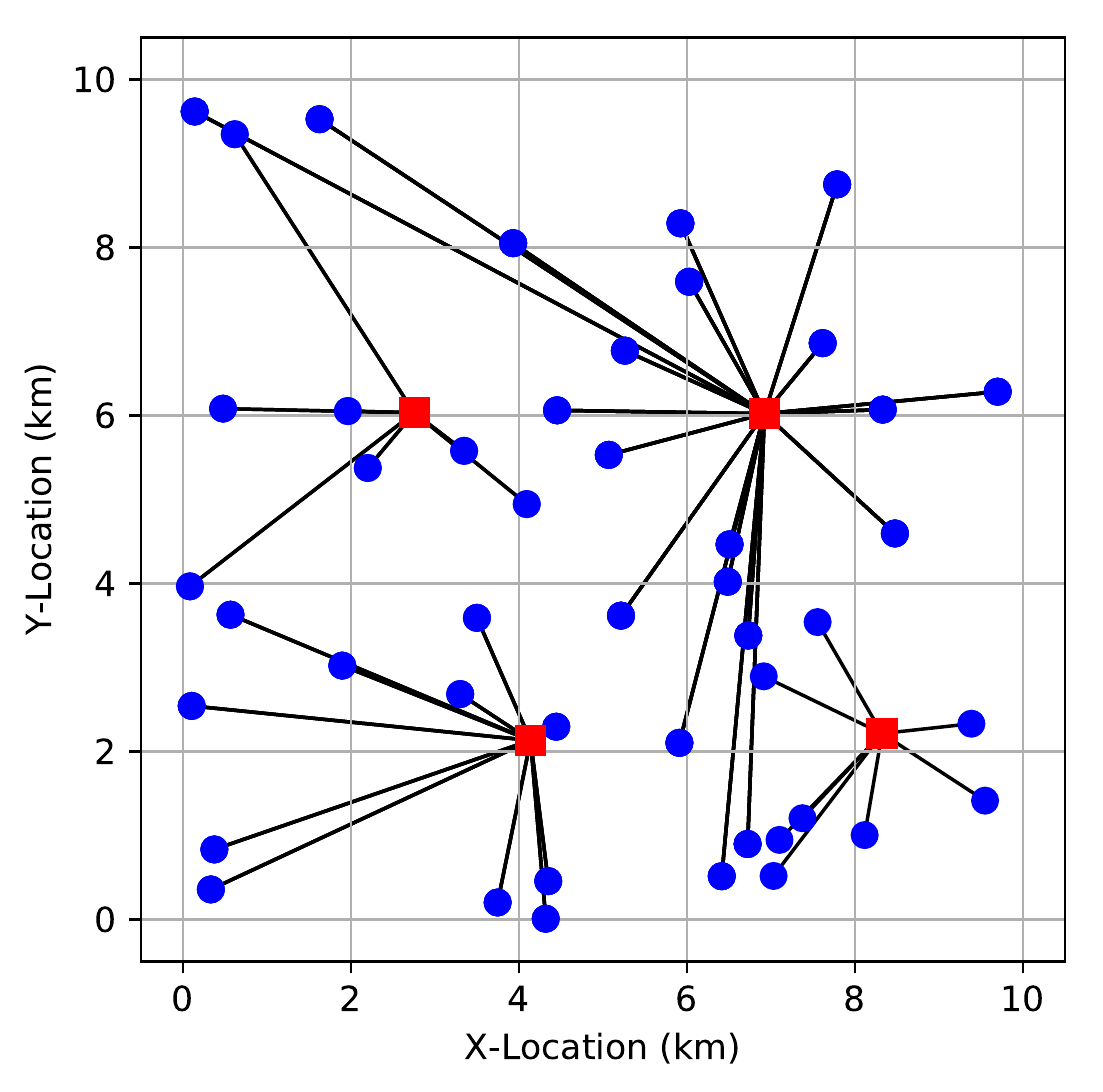}}\end{center}\vspace{-15pt}
\caption{A network topology for $\lambda = 0.1$.}
\label{conf1fig}
\end{figure}

\begin{figure}[htbp]
\begin{center}\scalebox{0.6}{\includegraphics{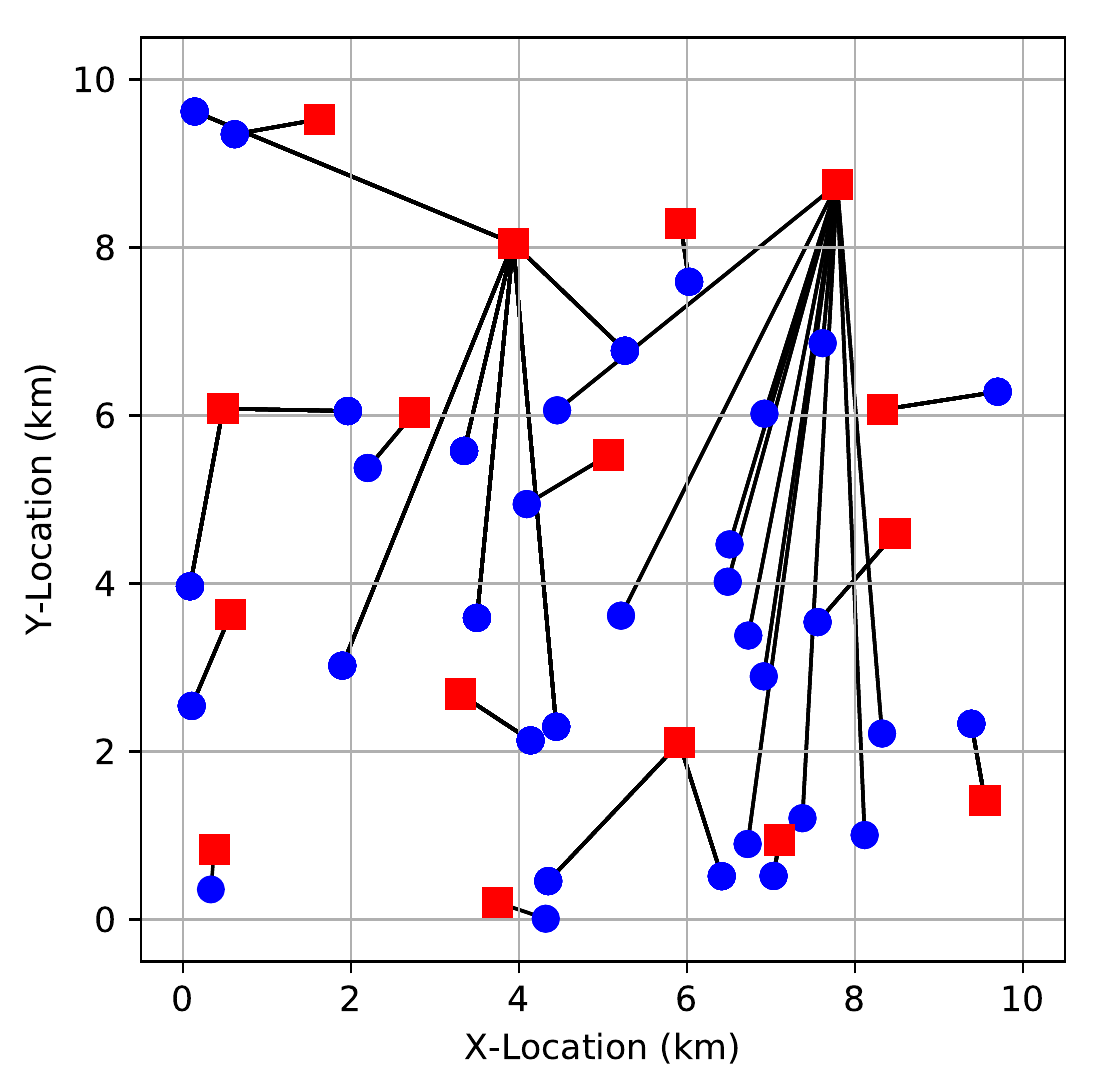}}\end{center}\vspace{-15pt}
\caption{A network topology for $\lambda = 0.5$.}
\label{conf2fig}
\end{figure}

The topology in Fig. \ref{conf2fig} achieves a coverage of $0.9939$ with computation rate $0.2396$; the corresponding point is also marked with a red cross in Fig. \ref{paretofig}. We can observe that there are now $16$ master nodes and $34$ worker nodes. The topology is also highly irregular in the sense that the smallest cluster has size $2$, and the largest cluster, whose master is located at around $(7.9,8.7)$, has size $13$. At first sight, it may seem that distributing the workers more evenly among the masters will result in a better performance, and that the topology in Fig. \ref{conf2fig} is thus strictly suboptimal. However, a careful investigation reveals that the topology in Fig. \ref{conf2fig} is a natural consequence of the coverage constraints and the specific node locations. In fact, let us attempt to construct a topology that achieves the same coverage of $0.9929$ with maximal computation rate. The very high coverage constraint implies that almost all locations on the area of interest should be covered. In this context, the node at $(7.9,8.7)$ is the only node that can cover the top right portion of the area of interest. Hence, (the node at) $(7.9,8.7)$ should be a master node. Similarly, the node at $(6,8.1)$ should be a master node as well. What makes $(7.9,8.7)$ particularly unlucky is that it has no close neighbors to aid its computations. The node at $(7.9,8.7)$ thus has to establish many long-range connections to compensate for the lack of short-range help, resulting in a topology similar to that in Fig. \ref{conf2fig}. The situation is similar for the node at $(4,8)$.

Finally, let us also note that we have also obtained the Pareto frontiers and the corresponding topologies for different realizations of node locations. For the same path loss exponent, the Pareto frontiers for different node realizations were nearly identical with only minor differences. This is due to the ``natural'' averaging out provided by our consideration of a relatively large number of networking nodes. We will report specific results elsewhere due to space limitations.


\section{Conclusions}
\label{secconclusions}
We have formulated and analyzed the tradeoff between the coverage and computation performances of wireless networks. We have shown that finding the optimal tradeoff is an NP-complete problem, and introduced an algorithm to find a locally-optimal solution. Many different variants and extensions of the problem formulated here can be studied. For example, one can consider the tradeoffs in the presence of node mobility,  different models of coverage, or fading scenarios.

\section*{Acknowledgement}
This work was supported in part by the NSF Award CCF-1814717, and in part by an award from the University of Illinois at Chicago Discovery Partners Institute Seed Funding Program.

\bibliography{main} 
\bibliographystyle{IEEEtran}

\end{document}